\documentclass{easychair}
\usepackage{doc}

\usepackage{url}

\title{Two-tier blockchain timestamped notarization with incremental security}

\author{Alessio Meneghetti\inst{1}\and
Armanda Ottaviano Quintavalle\inst{2} \and
Massimiliano Sala\inst{1}\and
Alessandro Tomasi\inst{3}}
\institute{Department of Mathematics, University of Trento, Trento, Italy \and
LDSD, Department of Physics and Astronomy, University of Sheffield, Sheffield, UK \and
Security and Trust, Fondazione Bruno Kessler, Trento, Italy
}

\authorrunning{Meneghetti et al.}

\titlerunning{Blockchain timestamped notarization}
\usepackage{graphicx}

\usepackage[utf8]{inputenc}
\usepackage[english]{babel}
\usepackage{csquotes}
\usepackage{amsmath, amsthm, amsfonts}

\usepackage{braket}
\usepackage{graphicx}

\usepackage{xcolor}
\usepackage{hyperref}

\newtheorem{theorem}{Theorem}[section]

\newcommand{\cat}{~{ \Vert}~}
\newcommand{\hash}[1]{h\left(#1\right)}
\newcommand{\sign}[2]{\sigma_{#1}\left({#2}\right)}
\newcommand{\Tree}{\mathcal{T}}
\newcommand{\Rho}{\mathrm{P}}
\newcommand{\Beta}{\mathrm{B}}
\newtheorem{algorithm}{Algorithm}
\begin{document}

\maketitle    
\begin{abstract}
Digital notarization is one of the most promising services offered by modern blockchain-based solutions. We present a digital notary design with incremental security and cost reduced with respect to current solutions.

A client of the service receives evidence in three steps. In the first
step, evidence is received almost immediately, but a lot of trust is
required. In the second step,
less trust is required, but evidence is received seconds later. Finally, in the third step evidence
is received within minutes via a public blockchain.
\end{abstract}

\section{Introduction}

The solution here described was commissioned to provide a customer in the financial sector with evidence to corroborate its statement of the integrity, authenticity, and existence at a given time of its data. This is closely related to the problems known as \textit{secure timestamping} and \textit{notarization}. 
The commission was made with the very specific requirement that it should make use of a privately run blockchain-based service anchored to a public blockchain, but at the same time still be capable of working without the private ledger, if it should cease to operate. We find it interesting to discuss how this design challenge can be met, and what security guarantees it can offer.

Digital timestamping and blockchain have been linked from inception. The bitcoin whitepaper \cite{Nakamoto_08} explicitly cites the linked timestamping work by Haber and Stornetta \cite{Haber_Stornetta_91} with Merkle trees \cite{Merkle_87} as efficiency improvement \cite{Bayer_Haber_Stornetta_91}.

Other blockchain-based timestamping solutions have been recently proposed. The authors of \cite{Gipp_Meuschke_Germandt_15} propose to commit aggregate data hashes to a bitcoin transaction. On the other hand, the authors of \cite{Yuefei_Gao_17} report on existing solutions that use the data hash as a bitcoin address to which to spend a transaction (BTProof, now unavailable), or embed a custom string and a data hash in the transaction's data field (Proof of Existence \cite{proof_of_existence}); the authors themselves extend the former by using three addresses per transaction: one encoding the name, one encoding metadata, and one encoding the data hash itself.

As pointed out in \cite{Matsuo_17}, there are several layers on which blockchain provides guarantees - of consistency, security, etc.\ - each with potentially different tools to check that these guarantees hold. The present solution does not offer a programming language for smart contracts and makes no specification as to the networking or consensus protocols.

The novel aspect of our solution is that a client of the service receives evidence in three steps. In the first
step, evidence is received almost immediately, but a lot of trust is
required. In the second step,
less trust is required, but evidence is received seconds later. Finally, in the third step evidence
is received within minutes via a public blockchain. We achieve our results thanks to the interaction between two blockchains, one of which is public.

\par\medskip

After some notation and preliminaries in Section \ref{sec:lit_review}, we give a semi-formal specification of the solution in Section \ref{sec:description}. Sections \ref{sec: proofs} and \ref{sec: advanced attacks} contain our security proofs. Section \ref{sec: comments} contains a discussion about possible DOS attacks, while Section \ref{sec: conclusions} hosts our conclusions.

\section{Notation and preliminaries}
\label{sec:lit_review}

Let time be measured in intervals $[a,b)=\{t\in\mathbb{R} \mid a\leq t < b\}$, for any $a<b\in \mathbb{R}$.
Let $\sign{A}{}$ be a digital signature algorithm using the private key of data owner $A$; if the owner can be determined unambiguously, we may sometimes omit $A$ for clarity. We assume $\sigma$ to be resistant to impersonation attacks: an attacker D cannot obtain $\sign{D}{m}$ from $\sign{A}{m}$. We also assume $\sigma$ to be unforgeable. For example, ECDSA satisfies both security properties under the usually-assumed hardness of the DLOG problem in (strong) elliptic curves.

Let $\cat$ denote string concatenation. As a shorthand when dealing with containers of data with attached a signature of the contents, e.g.\ block headers, we will commonly employ expressions such as
\begin{equation}\label{std}
s =  T \cat \sigma(s),
\end{equation}
where $T$ are the contents, the signature is computed as
\begin{equation*}
    \sigma(T \cat \underline{0})
\end{equation*}
with $\underline{0}$ a string of zeroes of the same length as $\sigma(T \cat \underline{0})$, and finally the container $s$ is composed by replacing $\underline{0}$ by $\sigma(T \cat \underline{0})$. The shorthand $\sigma(s)$ is employed to indicate that the signature refers to the entire contents, without writing them explicitly at length.

Throughout this paper we denote with $\eta_{A}(m)$ a public-key encryption of message $m$ with public key of actor $A$; with $\Tree\set{j}$ a Merkle tree of a list of items $j$; and with $\hash{m}$ a cryptographic hash function, in the sense of \cite{Rogaway_Shrimpton_04}, of message $m$. 
In particular, the collision resistance of $\hash{}$ is required in order to deduce that the forgery of a path in the Merkle tree would imply a collision.

The seminal work on digital timestamping is due to Haber and Stornetta \cite{Haber_Stornetta_91}, on which widely used standards for trusted timestamping today are based, such as RFC 3161 \cite{RFC_3161}, ANSI X9.95, and ISO 18014 \cite{ISO_Timestamping}.

We present the following two algorithms to recall these important methods and establish a common notation, though we do not make use of them directly in our solution.

\begin{algorithm}[Trusted authority timestamping]

	In \cite{Haber_Stornetta_91}, the client holding data $d$ sends its hash $\hash{d}$ to a trusted timestamping authority $\mathcal{A}$, which returns a signed statement $\tau$ of the time of receipt, $t$:
	\begin{align}
		\tau	& = \hash{d} \cat t \cat \sign{\mathcal{A}}{\tau} .
	\end{align}

	RFC 3161 is substantially the same scheme, but an additional hashing step is introduced:
	\begin{align}
		\tau	& = \hash{\hash{d} \cat t}\cat  \sign{\mathcal{A}}{\tau}
	\end{align}

\end{algorithm}

\noindent These schemes place all trust in the hands
of the authority $\mathcal{A}$, but in practice  trust is distributed among several stakeholders with successive
timestamps.

By using hash trees, more sophisticated timestamping algorithms have been developed (see e.g. \cite{Bayer_Haber_Stornetta_91,Haber_Stornetta_97}).

\begin{algorithm}[Tree-linked timestamping]
	At step $k$, define a time interval $[t_{k-1},t_{k})$. The server $\mathcal{A}$ collects all requests $\Theta_k=\{\theta_{k,i}\}_{i}$ received in the time interval, and builds the Merkle tree $\Tree(\Theta_k)$. We denote its root with $r_k$ and we call it \emph{interval root}. 

	The interval roots are linked together by the following rule, thus forming a chain of hashes $\{R_{\ell}\}_{0\leq \ell \leq k}$:

	\begin{align*}
	R_0&= 0 \\
	R_{\ell}& = \hash{R_{\ell-1} \cat r_{\ell}} 
	\end{align*}

	The values $\{R_{\ell}\}_{0\leq\ell\leq k}$ are placed in a widely available repository. The server then returns to each requester a receipt with the time $t_{k}$, along with the path in the Merkle tree from the requester’s leaf up to the value $R_k$.

\end{algorithm}

The authors of \cite{Haber_Stornetta_97} themselves point out that the integrity of the public repository of root hashes is the only requirement on which the authenticity of a document with receipt relies.

\section{Solution description}
\label{sec:description}
We now describe our solution. We do not require that blocks are created at fixed-time intervals, but we require a time division in intervals. To be more precise, 
since each block hosts the time of its  creation, we can consider time intervals using the index $k$ ($k\geq 0$), so that interval $k$ corresponds to $[t_{k-1},t_k)$ for $k\geq 1$ and interval $0$ corresponds to time $t < t_0$. 
\subsection{Architecture}

The service handles interactions among the following participants:
\begin{itemize}
	\item	The clients of the service wish to provably record the existence of data $d$ at a given time. $C$ denotes an arbitrary client.
	\item Some service nodes check proposed transactions for validity, provide evidence of record to clients, and maintain a blockchain, which we will call \textit{proxy blockchain}. $N$ and $M$ denote nodes. At block creation time, one of the service nodes acts as committing node and thus prepares a block, which is submitted to the other service nodes for acceptance.
	\item	One auxiliary node $A$ commits further evidence to a public ledger $L$, used as a reference clock and trusted timestamping service, and monitors client transaction activity. $A$ never acts as a service node. 
\end{itemize}

The service relies on a trusted certificate authority to provide the public key infrastructure necessary to both identify the nodes with permission to participate in the service and provide the ability to digitally sign documents. 
We assume that the proxy blockchain is at least permissioned if not private, and that it is a robust ledger in the sense of \cite{Garay_Kiayias_Leonardos_15}, i.e.\ guaranteeing persistence and liveness. We do not address the backbone  protocol and the consensus algorithm, in particular we make no specification as to how the information of transactions propagates, or how the node that shall create the next block is chosen, or how the other service nodes accept its proposed block.

\par\medskip

The main function of our service is delivered by enabling a client $C$ to prove existence of some data at commit time $t$, based on some evidence. We provide evidence in the form of digitally signed receipts and blockchain and $C$ receives evidence with \emph{incremental trust}, in three successive steps, as sketched below: 
\begin{enumerate}
\item Evidence issued by the service node receiving data.
	\item	Evidence issued by the service node that create blocks in the proxy blockchain.
	\item	Evidence issued by the auxiliary node $A$ and hosted by public blockchain $L$:  $A$ issues some evidence, which is partially stored in a public repository $L$, such as the bitcoin network.
\end{enumerate}
In the first step, $C$ receives a first signed receipt. In the second step, $C$  receives a second receipt and is able to access some evidence on the proxy blockchain.
In the last step, $C$ receives a final receipt and is able to access some evidence on a \emph{public} blockchain.
We speak of incremental trust because the probability of $C$ colluding with the other actors decreases significantly from one step to the next.

\subsection{Block creation and Node receipt issuance}
We now describe in detail the first phases of our system. We assume that all clients are operating in time interval $k$.
\par\medskip
\noindent\textbf{Transactions}\\
Each client $C$, identified by a digital identity $\iota_C$, creates a self-signed statement $\tau$  and sends it to one of the nodes $M$ for validation, which becomes in notation (\ref{std})
\begin{align}	\label{eqn:transaction}
	\tau	& = \sign{C}{d} \cat t\cat \iota_C \cat \sign{C}{\tau} 
\end{align}
This statement $\tau$ is analogous to a transaction in popular blockchain solutions, so we will call it \textit{transaction}. $M$ checks the signature in $\tau$ for validity and the claimed time $t$ (i.e. $t> t_{k-1}$). If the check is successful, $M$ broadcasts $\tau$ to the other
nodes and sends  $\sigma_M(\tau)$ back to $C$, which acts as the \emph{first receipt}.  We call $M$ the \textit{validator} of transaction $\tau$.

\bigskip

\noindent
\textbf{Block creation}
\\
The next committing Node $N$ then constructs a block by the following procedure. We assume that $k$ blocks have already been created, with block $k-1$ being the last created\footnote{Block 0 can be made with obvious modifications.}, meaning that all the following actions by $N$ are implicitly related to time interval $[t_{k-1},t_{k})$.

Let $T_C$ be the set of valid transactions originated by $C$ and received\footnote{This set might be smaller than the set of all transactions issued by $C$} by $N$ (via validator nodes), ordered according to a predefined rule\footnote{For example, according to the order defined by the integer representation of $\sign{C}{d}$.}.
Let $\Tree(T_C)$ be the Merkle tree of $T_C$, and let $R_C$ be its root. 
\\
Let $\mathcal{C}$ be the set of all Clients that originated transactions received by $N$, which we call \textit{transacting Clients}. Node $N$ calculates the Merkle tree  $\Tree\left(\left\{R_C\right\}_{C\in \mathcal{C}}\right)$, whose root is called $\Rho$. We refer to $R_C$ as the \textit{Client root} and $\Rho$ as the \textit{block root}.

We now define block $\Beta_k$ constructed by Node $N$.
The block will contain a \textit{header}, a list of \textit{summary transactions} (one per transacting Client), and a \textit{phantom part}.
\\
The summary transaction for $C$ contains a public encryption and is as follows:
\begin{align*}
\eta_{A}( \iota_C ) \cat R_C .
\end{align*}

The header $\mathcal{H}_k$ of $\Beta_k$ is, in notation (\ref{std}),
\begin{align}	\label{eqn:block}
	\mathcal{H}_k	& = \hash{\mathcal{H}_{k-1}} \cat k \cat t_k \cat\Rho \cat \iota_N \cat \sign{N}{\mathcal{H}_k} ,	
\end{align}
where $t_k$ is stated by $N$ as the creation time of $\Beta_k$.

The phantom part is the list of all transactions referred by the summary transactions, that is $\cup_{C\in\mathcal{C}}T_C$. This transaction list is called \textit{phantom part} because it is a part of the block which is visible only to the Nodes, and so invisible to the Clients. 
\bigskip

\noindent
\textbf{Receipts}\\
Once $N$ creates the block $\Beta_k$, $N$ issues a receipt $\rho_C$ to each transacting Client $C$, which in notation (\ref{std}) is written
\begin{align}	\label{eqn:receipt_Node}
	\rho_C	& = T_C \cat R_C \cat \pi_C \cat \sign{N}{\rho_C} ,
\end{align}
where $\pi_C$ is the shortest path in the Merkle tree from $R_C$ to $\Rho$.

\subsection{Public ledger and Auxiliary receipt issuance}
Let $m$ be a fixed number $m\geq 2$. Every $m$ blocks, node $A$ interacts with the public blockchain. Let $k_0$ be the last time interval in which this happened, and call \textit{anchorage block} a block corresponding to one of these interactions.
At the end of time interval $k_0+m$, another anchorage block is created, therefore $A$ collects the ordered list of Merkle roots and block hashes of blocks $k_0+1,\ldots, k_0+m$ and uses these as $2m$ leaves of an auxiliary Merkle tree $\Tree_{A}$,
\begin{align}
	\Tree_{A}=\Tree\left(\set{\Rho_k, \hash{\mathcal{H}_k}}_{k_0+1\leq k\leq k_0+m}\right)
\end{align}
with root $\mathcal{R}_{k_0+m}$, referred to as the \textit{auxiliary root}.

The data committed to the public ledger\footnote{The commitment of these data to the public ledger may require more than one transaction, according to the public-blockchain transaction format.} will be
\begin{align}
	\mathtt{pub\_data}=\iota_A \cat k_0 \cat m \cat \mathcal{R}_{k_0+m} \,.
\end{align}
Finally, the auxiliary node issues an auxiliary receipt $\underline{\rho}_C$ to every client transacting in the intervals $k_0+1\leq k\leq k_0+m$. Each such Client will already be in possession of a set of receipts \eqref{eqn:receipt_Node}, which contains a set of blocks roots inside the paths $\pi_C$'s. This set of block roots is a subset of the leaves of the tree with root $\mathcal{R}_{k_0+m}$. Therefore, $\underline{\rho}_C$ will contain the shortest path $\underline{\pi}_C$ in $\Tree_{A}$ required for $C$ to recompute $\mathcal{R}_{k_0+m}$:
\begin{align}\label{underrho}
	\underline{\rho}_C	& = k_0 + 1 \cat k_0 + m \cat \mathtt{pub\_data} \cat v \cat \underline{\pi}_C\cat \sign{A}{\underline{\rho}_C} \,,
\end{align}
where $v$ is the address in the public blockchain of the transaction containing $\mathtt{pub\_data}$.

\section{Proof of notarization}
\label{sec: proofs}
The system that we have described in the previous section may appear
unnecessarily complicated. After all, if there exists a ``good''
timestamping server, people could just use it and get its signature in
return. But herein lies the problem, because for a timestamping server to be
``good'' it needs to be secure, reliable, approachable - in the sense that
it is easy to communicate with it, both in bandwidth and in
permissions - and cheap to use. While features like reliability,
connectivity, and cost can be relatively easy to estimate, security
remains much more difficult to evaluate. Indeed, we are not aware of
any timestamping service on the Internet that presently satisfies all
these properties, especially security.

The only system that might provide reasonable security is a public
blockchain, such as the Bitcoin, and it would easily provide also
approachability and reliability (in particular, avoiding the risk of a single point of failure). However, at present, the cost of
transactions on a public blockchain is very high, making its direct use for storing proofs of documents infeasible. Therefore, many
competing solutions have been proposed, whose general aim is to
collect information on many documents - typically in hash form - and
create paths of hashes linking each document to the final digest
released on the public blockchain, e.g., Eternity Wall \cite{Eternity_wall}, Factom \cite{Factom_white_paper}, and
Guardtime \cite{Guardtime}.
These solutions must give their users some sort of receipt, allowing them
to reconstruct the hash path and prove the existence of their documents.

Although our solution may appear similar, we aim at something more: we want to
give our users
incremental security. In the next sections we will describe our security claims and provide proofs.

\subsection{Security claims}
\label{sec:goals}

The solution in Section \ref{sec:description} builds on the basic premise of digital notarization of a document. 
Each client correctly interacting with our previous system (in a
correct implementation) may be seen as a notary making a statement of the existence of data $d$ by adding their digital signature, $\sign{C}{\hash{d}\cat t}$. We are not concerned with the kind of information carried by $d$. Indeed, the data may or may not be a document signed by parties entering a contract, and their handwritten signatures may have been added on a paper or digital copy; we here emphasize that the only digital statement of \textit{authenticity} by digital signature is the Client's.

We assume throughout that the blockchain is a \emph{robust ledger} in the sense of \cite{Garay_Kiayias_Leonardos_15}, i.e.\ guaranteeing persistence and liveness. Since no specification of consensus mechanism is made here by design, we do not consider the question of whether the blockchain offers sufficient guarantees of consistency, crash fault tolerance, or security against collusion.

Let us consider the first step of our incremental security, which is
the first interaction of a client $C$ with our system. Let us call ``Tom''
someone who will come and will not believe in $C$'s claims about the
existence of its claimed documents at the claimed time. $C$ hopes to be able
to convince Tom by using our protocol.

At the start, 
$C$ sends the signature of a document to the node network,
encapsulated in the transaction $\tau$  \eqref{eqn:transaction}, and a node $M$
receives it.
When $M$ sends back to $C$ the signature
$\sign{M}{\tau}$, $M$ is actually giving $C$ the \underline{first evidence} that
$C$ can show to Tom
about its good faith. $C$ is therefore immediately - say, in a few seconds - in possession of 
a receipt claiming the existence of its documents at the claimed
time.

Whether Tom trusts this claim depends on whether Tom trusts $M$, specifically. This is
equivalent to trusting a single timestamping server and can be
modeled simply as follows.
\begin{theorem}
If $M$ is trusted, then the data claimed by $C$ existed at the claimed
time and were known to $C$.
\end{theorem}
\begin{proof}
This comes from the unforgeability property of the signature
$\sigma_C$ and $\sigma_M$.
\end{proof}

In the second step, to increase trust we need to
increase the number of entities working for the system, but making
sure they arrive at an agreement
on the documents. Nowadays, this can be achieved with a private
blockchain. We require it to be private so that Tom knows all service nodes
and can decide whether he trusts them.

Observe that we have not specified how consensus is reached in
the proxy blockchain. It could be that a majority of nodes is needed,
or that all nodes must confirm $N$'s proposed block, or some other more
complex strategy. It does not matter, as long as Tom agrees that the consensus algorithm is trustworthy.
What matters is that $C$ has collected the new block header and that
he has received a receipt $\rho_C$ \eqref{eqn:receipt_Node} from $N$.
With this \underline{second evidence}, Tom will agree on the following
\begin{theorem}
If the new block has been generated with a trusted consensus
algorithm, then (at least at time $t_{k}$) the data claimed by $C$ existed and were known to $C$.
\end{theorem}
\begin{proof}
All the service nodes that reached consensus have seen $T_C$, so by
signalling their agreement to the new block they agreed that the
transactions from $C$ were valid and,
in particular, that the transactions were sent before the creation time
$t_k$ of the block. Crucially, Tom has not seen the phantom part of
the block, but he does not need it.
Indeed, from $\mathcal{H}_k$ Tom can get $t_k$ and $\Rho$. From
$\rho_C$ he gets $T_C$ and $\pi_C$, which shows the validity of the
hash path.\\
$C$ could not have forged the hash path $\pi_C$ due to the hash
security properties.
\end{proof}
An obvious question might arise when looking at the above proof: why
keep a phantom part?
This need  arises from privacy reasons: we do not want any client $C$
to see the transactions coming from another client $C'$. Of course,
this goal could
be reached with e.g.\ cryptography, but we take advantage of the use of a private blockchain
to avoid more complicated features.

Thus $C$ obtains within a short period of time
some evidence on its documents that provides a much higher confidence
for Tom:
it is one thing to compromise or collude with a single participant $M$, and another to
organize a collusion among the service nodes, including the miner $N$.

\par\medskip

In the third and last step, if Tom does not trust the
proxy blockchain, we will assume that he trusts the anchored public
blockchain. Again, this trust means that,
whatever consensus algorithm employed and whatever participants
involved, the anchored public block was issued at a time prior to
$\bar t$.

Assuming that $C$ has collected the new public block 
and the receipts $\underline{\rho}_C$ and $\rho_C$ ($\rho_C$ refers to time interval $k$), Tom will then agree on the following.
\begin{theorem}
If a new public block has been generated in a public blockchain
by a trusted consensus algorithm, then (at least at time $\bar t$) the data claimed by $C$ existed and were known to $C$.
\end{theorem}
\begin{proof}
The new public block contains a public transaction containing
$\texttt{pub\_data}$. Since the public nodes have reached consensus and
Tom trusts the public blockchain, he
will trust that \texttt{pub\_data} existed before $\bar{t}$.
Indeed, from $\texttt{pub\_data} $ Tom can get $\mathcal{R}_{k_0+m}$. From
$\underline{\rho}_C$ he gets $\underline{\pi}_C$, which shows the
validity of the hash path from $R_{k_0+m}$ to $\Rho_K$. 
From
$\rho_C$ he gets $\pi_C$, which shows the
validity of the hash path from $\Rho_K$ to $R_C$.
From $\rho_C$ he also extracts $T_C$, which contains the transaction with the claimed data, and can check its validity by recomputing $R_C$.

$C$ could not have forged the hash paths $\underline{\pi}_C, \pi_C$ or the tree root $R_C$, without incurring in a hash collision. 
\end{proof}
Observe that in this last step Tom does not need to put any trust in
the proxy blockchain, because he can verify by himself the related
hash chains.

\par\medskip

Obviously, in this third step of incremental trust Tom does feel very
sure about $C$'s claim, but $C$ obtains all the needed \underline{third evidence} only
\textit{after} the public block creation, which will probably last some minutes
and its timestamping claim can only be up to $\bar{t}$ rather than its
claimed $t$.


\section{Attacks with widespread collusion}
\label{sec: advanced attacks}
In Section \ref{sec: proofs} we showed how Tom can be convinced by $C$ in
different trust scenarios.
To convince Tom, $C$ needed some valid evidence from the system.
But the system might decide not to release such evidence
and try to obtain some advantage for itself.
We will consider now attack scenarios where the system does not
interact correctly with $C$.
We will assume that all node services (including miners and auxiliary
nodes) are malevolent.

The three scenarios we are considering are: ``Fake Owner'', ``Ghost document: proxy version'', ``Ghost document: public version''. In the  ``Ghost document'' scenarios also the client $C$ is malevolent.

\par\medskip

\textbf{Fake Owner}\\
Client $C$ sends a transaction $\tau$ \eqref{eqn:transaction} at time $t$. The system does not acknowledge $\tau$, and instead 
provides a colluding client $D$ with evidence enough to claim that $D$
knew data $d$ at time $t$.
\begin{theorem}
The Fake Owner attack fails.
\end{theorem}
\begin{proof}
First, the malevolent nodes need to create a fake transaction
$\tau'$. We do not need to model what they will do next, because we
claim
it is impossible to create it. Indeed, $\tau'$ would have the form in notation (\ref{std})
$$
   \tau'=\sigma_D(d) || t || \iota_D || \sigma_D(\tau')
$$
In particular, it would contain $\sigma_D(d)$. However, we assumed
that our signature algorithm was resistant to impersonation attacks,
and so
this cannot happen.
\end{proof}

\textbf{Ghost document: proxy version} \\
After a new proxy block $\mathcal{H}_k$ is created at time $t_k$ - but
before an anchorage block, i.e. $k \ne m \lambda$ for any $\lambda$ -
$C$ colludes with all service nodes
and inserts a new transaction claiming that $C$ knew data $d'$ at time
$t' < t_k$.

This attack may work, because the nodes can decide to:
\begin{itemize}
\item discard block $k$ and all existing successive blocks in their blockchain,
\item recompute $T_C$ (to include $\tau'$) and all relevant
information in $\mathcal{H}_k$,
\item recompute $\rho_C$ (to include the new $T_C$), recompute all
receipts for the other clients (and send them);
\item recompute the headers of the successive blocks (to have valid
hash pointers) and resend them to the clients.
\end{itemize}
The other clients might complain at seeing the change in past block
headers, but they may be satisfied when they receive valid blocks
and valid receipts. If this attack is performed rarely, the clients
(and Tom) may be induced into believing that the block updates are due
to some software problems
rather than malice.

\par\medskip

\textbf{Ghost document: public version} \\
After a new proxy block $\mathcal{H}_k$ is created at time $t_k$ (an
anchorage block, so $k = m \lambda$ for some $\lambda$) and the anchor
has been created at time $t'$, $C$ colludes with all service nodes and
inserts a new transaction in $\mathcal{H}_k$ claiming that $C$ knew
data $d'$ at time $t' < t_k$.
\begin{theorem}
A ghost document in public version cannot be created.
\end{theorem}
\begin{proof}
Since the anchor has been created, \texttt{pub\_data} is in the public
blockchain. To be able to claim knowledge of data $d'$, $C$ needs a
valid hash path
pointing to the root
corresponding to the new $T_C$, so that he could use it to replace $\underline{\pi}_C$. However, this path must end in
$\mathcal{R}_{k}$, which is immutable in the public blockchain, and so
it is impossible to forge another path due to security properties of
the hash function.
\end{proof}

\section{Some comments on DOS attacks}
\label{sec: comments}

In the past section we do not investigate scenarios when malevolent actors of
the system want to mount a DOS (Denial Of Service).
We now examine this situation. There are three possible attackers:
a validator node, the auxiliary node $A$ and the PKI's CA.

\textbf{Validator node} If a malevolent node $M$ receives a transaction $\tau$ from a client, $M$ can decide to ignore it. In this case, $C$ would notice that something went amiss and $C$ would try to contact another node $N$. This kind of DOS is dangerous only if the client's communication
with the system is limited to a group of colluding nodes.\\
A malevolent $M$ could do worse than simply dropping $\tau$: $M$ could send the first receipt
$\sigma_M(\tau)$ back to $C$ and avoid broadcasting it to the other service nodes.
In this way, $C$ is tricked into thinking that $M$ is behaving honestly. However, $C$ is expecting to receive also the second receipt $\rho_C$ in a while. If this does not happen, $C$ will know something is wrong and it will then interact with other nodes. On the other hand, if $C$ receives $\rho_C$ and $\tau$ has not been used to construct the relevant hashes, then
again $C$ will notice something is wrong.

\textbf{Auxiliary node} $A$ cannot modify the $ H_k$'s to construct its Merkle tree (and thus compute a valid $\mathtt{pub\_data}$), because $A$ cannot sign impersonating one of the service nodes.
However, $A$ may avoid to insert some of the $H_k$'s, effectively removing from the auxiliary tree all the transactions received by the clients in the chosen intervals. This DOS attack by $A$ is easily spotted by all nodes (and all clients) when the anchoring happens, because it will be impossible to reconclie the issued auxiliary receipts and the other receipts held by the clients.

\textbf{Certification Authority} We assume in our system that the CA is trusted, because if the CA were to issue certificates to malicious peers, and if it failed to revoke them, the system would be vulnerable to a majority attack. However, it could be that the CA itself is flooded by packets sent by DOS attackers. In this scenario, it may impossible for the system participants to check the validity of new data coming into the system, depending on the public keys held by each participant. Yet, assuming that honest peers will not validate transactions without a certificate revocation list being available, the validity of past transactions remains perfectly checkable by anyone having the relevant receipts (including Tom, if $C$ gives them to him), since they are enough to validate the data inserted in the public blockchain.

\section{Conclusions}
\label{sec: conclusions}

We have shown a two-tiered system of independent blockchains for secure timestamping that offers incremental levels of evidence to clients. We have examined under what assumptions the system may be deemed secure; in particular, we have seen that under the assumption of an honest certification authority, only denial of service attacks are feasible, and they are also immediately noticeable. The two-tiered system is designed to reduce the cost and increase efficiency of commitments to a slow and costly public blockchain, while at the same time still enabling clients to use their past evidence even if the intermediate blockchain solution were to cease being operational.
\\
While we are satisfied with our finding, we notice that our results hold in a blockchain having an indefinite but supposedly robust consensus algorithm. It would be interesting to investigate how our system could be effectively integrated in a blockchain enjoying a specific consensus algorithm, such as proof-of-work or proof-of-stake.

\section*{Acknowledgments}
Some partial results of this paper have been presented at the Euregio Blockchain Conference in 2018. The more advanced results presented here have been funded by the project MIUR PON ``Distributed Ledgers for Secure Open Communities".

\bibliographystyle{plain}
\bibliography{Refs}

\end{document}